\documentclass[a4paper,11pt]{article}

\usepackage[mathlines]{lineno}
\usepackage{amsmath}
\usepackage{amssymb}
\usepackage{amsthm}
\usepackage{graphicx}
\usepackage{enumerate}
\usepackage{cite}
\usepackage{fullpage}
\usepackage[usenames]{color}
\usepackage[nodayofweek]{datetime}
\usepackage{hyperref}

\allowdisplaybreaks

\newcommand{\Real}{\ensuremath{\mathbb{R}}}

\newcommand{\intr}{\mathop{\mathrm{int}}}

\newcommand{\Ball}{\ensuremath{\mathbf{B}}}

\newcommand{\proj}{\overline}
\newcommand{\fproj}{\proj{f}}
\newcommand{\VD}{\ensuremath{\mathsf{VD}}}

\newtheoremstyle{mytheorem}{3pt}{3pt}{\slshape}{}{\bfseries}{}{.5em}{}
\theoremstyle{mytheorem}
\newtheorem{lemma}{Lemma}
\newtheorem{theorem}{Theorem}

\newtheorem{corollary}{Corollary}

\theoremstyle{definition}

\makeatletter
\newbox\ProofSym
\setbox\ProofSym=\hbox{%
\unitlength=0.18ex%
\begin{picture}(10,10)
\put(0,0){\framebox(9,9){}} \put(0,3){\framebox(6,6){}}
\end{picture}}
\renewenvironment{proof}[1][Proof.]{\O@proof{#1}}{\O@endproof}
\def\O@proof#1{\trivlist
   \@topsep\z@\@topsepadd\smallskipamount%
   \@ifstar{\item[]}{\item[\hskip\labelsep\it #1 ]}}
\def\O@endproof{\hfill\copy\ProofSym\linebreak[3mm]\endtrivlist}
\makeatother

\def\denseitems{
    \itemsep1pt plus1pt minus1pt
    \parsep0pt plus0pt
    \parskip0pt\topsep0pt}
\begin{document}


\title{Computing a Minimum-Width Cubic and Hypercubic Shell%
\thanks{%
This work was supported by Kyonggi University Research Grant 2018.
}
}

\author{%
Sang Won Bae\footnote{%
Division of Computer Science and Engineering, Kyonggi University, Suwon, Korea.
Email: \texttt{swbae@kgu.ac.kr} }
}

\date{%
\today\quad\currenttime
}

\maketitle

\begin{abstract}
In this paper, we study the problem of computing a minimum-width
axis-aligned cubic shell that encloses a given set of $n$ points
in a three-dimensional space.
A cubic shell is a closed volume between two concentric and face-parallel cubes.
Prior to this work, there was no known algorithm for this problem in the literature.
We present the first nontrivial algorithm
whose running time is $O(n \log^2 n)$.
Our approach easily extends to higher dimension, resulting in an
$O(n^{\lfloor d/2 \rfloor} \log^{d-1} n)$-time algorithm
for the hypercubic shell problem in $d\geq 3$ dimension.
\\

\end{abstract}

\section{Introduction} \label{sec:intro}

The minimum-width circular annulus problem asks to find an annulus
of the minimum width, determined by two concentric circles,
that encloses a given set $P$ of $n$ points in the plane.
It has an application to the points-to-circle matching problem,
the minimum-regret facility location, and the roundness problem.
After early results on the circular annulus problem~\cite{w-nmppp-86,rz-epccmrsare-92},
the currently best algorithm that computes a minimum-width circular annulus
that encloses $n$ input points takes $O(n^{\frac{3}{2}+\epsilon})$ time~\cite{ast-apsgo-94,as-erasgop-96}.

Along with these applications and with natural theoretical interests,
the \emph{minimum-width annulus problem} and its variants
have recently been attained a lot of attention by many researchers,
resulting in various efficient algorithms.
Abellanas et al.~\cite{ahimpr-bfr-03} considered minimum-width rectangular annuli
that are axis-parallel, and presented two algorithms taking $O(n)$ or $O(n \log n)$ time:
one minimizes the width over rectangular annuli with arbitrary aspect ratio
and the other does over  rectangular annuli with a prescribed aspect ratio, respectively.
Gluchshenko et al.~\cite{ght-oafepramw-09} presented an $O(n \log n)$-time algorithm
that computes a minimum-width axis-parallel square annulus,
and proved a matching lower bound,
while the second algorithm by Abellanas et al.\@ can do the same in the same time bound.
If one considers rectangular or square annuli in arbitrary orientation,
the problem gets more difficult.
Mukherjee et al.~\cite{mmkd-mwra-13} presented an $O(n^2 \log n)$-time algorithm
that computes a minimum-width rectangular annulus in arbitrary orientation
and arbitrary aspect ratio.
The author~\cite{b-cmwsaao-18} recently showed that
a minimum-width square annulus in arbitrary orientation can be computed
in $O(n^3 \log n)$ time.

Despite of these recent progress and successful generalizations,
little is known about the high dimensional variants of the annulus problem.
For $d \geq 3$, the $d$-dimensional generalization of annuli is often referred to
\emph{shells} of a certain body of volume.
Mukherjee et al.~\cite{mmkd-mwra-13} showed that
a minimum-width shell of $d$-dimensional axis-parallel boxes (or hyper-rectangules)
can be computed in $O(dn)$ time.
For the minimum-width spherical or hyperspherical shells,
Chan~\cite{c-adwsecmwa-02} showed an $O(n^{\lfloor d/2 \rfloor + 1})$-time
exact algorithm, and some approximation algorithms are known~\cite{aahs-aamwas-00, c-adwsecmwa-02}.
However, to our best knowledge, there is no known result for
the cubic or hypercubic shell problem in the literature.
In fact, it is not difficult to apply Chan's approach and algorithm~\cite{c-adwsecmwa-02} for the hyperspherical shells to the hypercubic shells,
which implies $O(n^{\lfloor d/2 \rfloor + 1})$-time algorithm exact algorithm
that computes a minimum-width hypercubic shell enclosing $n$ points in $\Real^d$.
This in particular implies an $O(n^2)$-time algorithm that computes
a minimum-width cubic shell for $d=3$.

In this paper, we address the minimum-width hypercubic shell problem
in three or higher dimensions.
We first handle the three dimensional case,
and present a new algorithm that computes a minimum-width axis-aligned cubic shell
that encloses $n$ input points.
Our algorithm is based on a new approach which is different from that of Chan~\cite{c-adwsecmwa-02}, and takes $O(n \log^2 n)$ time in the worst case.
Next, we show that our approach can be extended to higher dimensions $d > 3$,
and present an algorithm that runs in
$O(n^{\lfloor d/2 \rfloor} \log^{d-1} n)$ expected time.

The rest of the paper is organized as follows:
We start with introducing some preliminaries in Section~\ref{sec:pre}.
After providing basic observations on hypercubes in $\Real^d$ for $d\geq 3$
in Section~\ref{sec:observations},
we discuss the case of $d = 3$ dimension in Section~\ref{sec:cubic_shell}
and present our algorithm that computes a minimum-width axis-aligned
cubic shell that encloses $n$ input points in $\Real^3$.
We then extend our approach and algorithm to higher dimensions in
Section~\ref{sec:hypercubic_shell}.
We finally concludes our paper with Section~\ref{sec:conclusions}.

\section{Preliminaries} \label{sec:pre}

In this section, we introduce some preliminaries for our discussions.
We consider the $d$-dimensional space $\Real^d$ for $d \geq 1$
with a standard coordinate system of $d$ axes,
namely, the $x_1$-axis, $x_2$-axis, \ldots, and $x_d$-axis.
For any point $p \in \Real^d$, its coordinates will be referred to
$x_1(p), x_2(p), \ldots, x_d(p)$ in this order,
so $p = (x_1(p), x_2(p), \ldots, x_d(p))$.
The \emph{$L_\infty$ norm} of $p \in \Real^d$, denoted here by $\| p\|$, is
defined to be
 \[ \|p \| := \sum_{i=1, \ldots, d} |x_i(p)|.\]
For any two points $p, q \in \Real^d$,
the \emph{$L_\infty$ distance} between $p$ and $q$ is $\| p - q\|$.
The \emph{$L_\infty$-ball centered at $p$ with radius $r \in \Real$},
denoted by $\Ball(p, r)$, is
the set of points $q \in \Real^d$ such that $\| q - p\| \leq r$.
A $d$-dimensional \emph{axis-aligned hypercube} is
a synonym to an $L_\infty$-ball in $\Real^d$.
In particular, an axis-aligned hypercube is called an \emph{interval} if $d=1$;
an axis-aligned \emph{square} if $d=2$; and an axis-aligned \emph{cube} if $d=3$.
The side length of a hypercube is twice its radius.
Throughout this paper, we only discuss axis-aligned hypercubes,
so we shall mean axis-aligned hypercubes without the adjective ``axis-aligned.''

\begin{figure}[tb]
\begin{center}
\includegraphics[width=.7\textwidth]{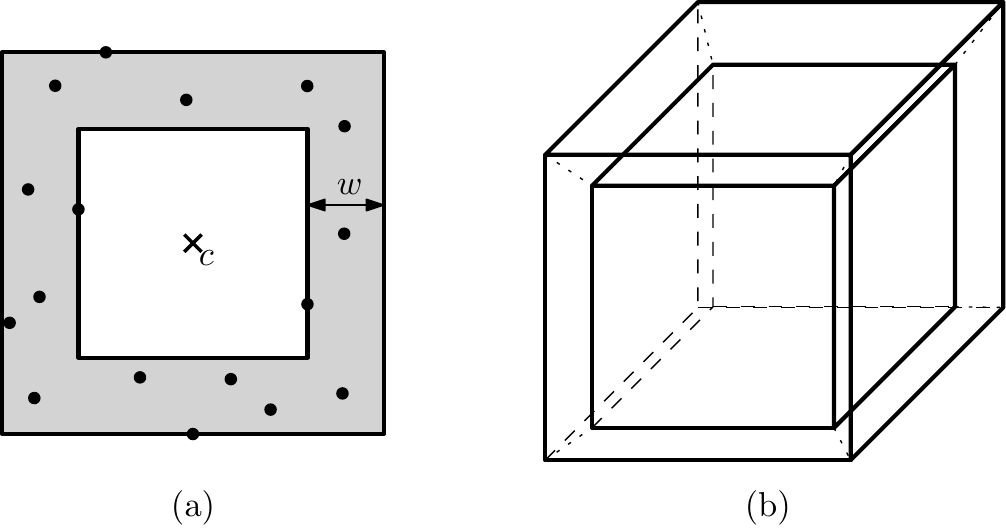}
\end{center}
\caption{(a) A minimum-width square annulus $A$ (shaded area) enclosing a set of $17$ points in the plane whose width is $w$ and center is $c$. Its outer and inner squares are drawn in thick black lines. (b) A cubic shell in $\Real^3$ whose outer and inner cubes are drawn in thick black lines.
 }
\label{fig:cubicshell}
\end{figure}

Two hypercubes are called \emph{concentric} if they share a common center.
A \emph{hypercubic shell} $A$ in $\Real^d$ is
the closed volume between two concentric hypercubes,
called the \emph{outer hypercube} $B$ and the \emph{inner hypercube} $B'$ of $A$,
respectively, where the radius of $B$ is at least that of $B'$.
Specifically, $A = B \setminus \intr B'$,
where $\intr B'$ denotes the interior of $B'$.
The \emph{width} of a hypercubic shell is the difference
between the radii of its inner and outer hypercubes.
A hypercubic shell is also called a \emph{square annulus},
in particular, for $d=2$, and a \emph{cubic shell} for $d=3$.
See \figurename~\ref{fig:cubicshell} for an illustration
of a square annulus in $\Real^2$ and a cubic shell in $\Real^3$.

The main purpose of this paper is to solve the
\emph{minimum-width hypercubic shell problem},
in which we are given a set $P$ of $n$ points in $\Real^d$ for $d \geq 1$
and want to find a hypercubic shell of minimum width that encloses $P$.
The problem is also called the \emph{minimum-width square annulus problem}
for $d=2$ and the \emph{minimum-width cubic shell problem} for $d=3$.

As introduced above, the minimum-width square annulus problem
can be solved in $O(n \log n)$ time in the worst case,
and its matching lower bound is also known~\cite{ght-oafepramw-09}.
The case of $d=1$ would be less interesting, while it is worth mentioning
for completeness.
For $d=1$, the problem is to compute two intervals of equal length that contain
$n$ given numbers $P \subset \Real$,
and it can be easily done in $O(n)$ time.
\begin{theorem} \label{thm:1d-2d}
 The minimum-width hypercubic problem can be solved in $\Theta(n)$ time for $d=1$
 and $\Theta(n\log n)$ time for $d=2$, both in the worst case.
\end{theorem}

In the following, we consider the problem for $d = 3$ and higher.
For the purpose, we need a basic geometry of cubes and hypercubes
enclosing the given set $P$ of points.
Throughout the paper, we shall say that
a facet of a hypercube or a box contains a point $p\in \Real^d$
if the facet or any face of less dimension incident to it contains the point $p$.
\begin{lemma} \label{lem:smallest_enclosing_cube}
 Let $P \subset \Real^d$ for $d \geq 2$ be a set of points, and $B$ be a hypercube
 that encloses $P$.
 Then, $B$ is a smallest enclosing hypercube for $P$ if and only if
 there are two parallel facets of $B$ such that
 each of them contains a point of $P$.
\end{lemma}
\begin{proof}
If there is no pair of parallel facets of $B$, each of which contains a point of $P$, then $B$ is certainly not of the smallest size.
Conversely, suppose that there are two parallel facets of $B$
containing a point of $P$ on each.
Let $p, p' \in P$ be these two points on the parallel facets of $B$.
Then the radius of $B$ is determined by $p$ and $p'$, $\|p-p'\|/2$.
On the other hand, any hypercube $B'$ enclosing $P$ should includes
these two points $p$ and $p'$,
so the radius of such hypercube $B'$ cannot be smaller than $\|p - p'\|/2$.
Hence, $B$ is a smallest enclosing hypercube for $P$.
\end{proof}

\section{Basic Observations on Hypercubic Shells} \label{sec:observations}
In this section, we observe some general properties of hypercubic shells
enclosing a set of points.
Let $d \geq 3$ be an integer, being a constant, and
$P$ be a set of $n$ points in $\Real^d$.

\begin{figure}[tb]
\begin{center}
\includegraphics[width=.7\textwidth]{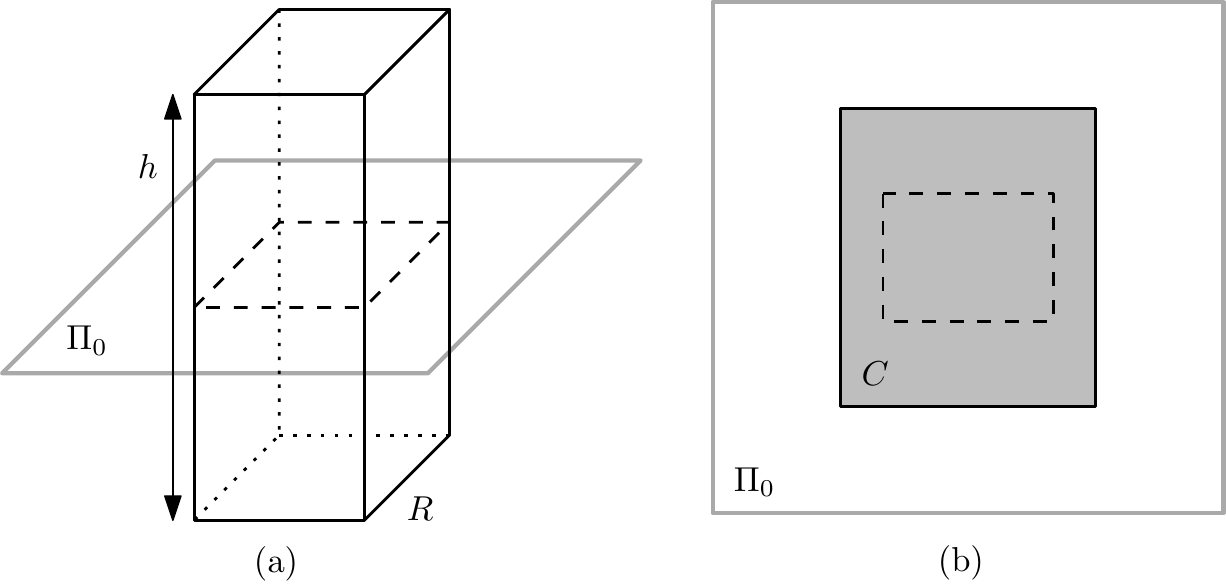}
\end{center}
\caption{Illustration for $d=3$ for (a) the smallest enclosing box $R$ for $P$ and the hyperplane $\Pi_0$ that halves $R$, and (b) the set $C$ of centers of all smallest hypercubes that enclose $P$, being a rectangle for $d=3$. Here, the dashed rectangle is the intersection $R\cap \Pi_0$.
 }
\label{fig:C}
\end{figure}

Let $R$ be the smallest axis-aligned box, or hyperrectangle, that encloses $P$,
that is,
\[R = [\min_{p\in P} x_1(p), \max_{p\in P} x_1(p)] \times [\min_{p\in P} x_2(p), \max_{p\in P} x_2(p)] \times \cdots \times [\min_{p\in P} x_d(p), \max_{p\in P} x_d(p)].\]
Let $h$ be the length of the longest sides of $R$, and
without loss of generality, we assume that
the sides of $R$ with length $h$ are parallel to the $x_d$-axis.
Consider the hyperplane $\Pi_0$ orthogonal to the $x_d$-axis that halves $R$.
Again, we assume that $\Pi_0$ contains the origin $o=(0,0, \ldots, 0)$,
i.e., $\Pi_0$ coincides the $x_1x_2\cdots x_{d-1}$-hyperplane,
which can be easily achieved by a translation of $P$ along the $x_d$-axis.
See \figurename~\ref{fig:C}(a) for an illustration for $d=3$.

We then consider any smallest axis-aligned hypercube $B$ that encloses $P$.
Let $C$ be the set of centers of all such smallest hypercubes that enclose $P$.
\begin{lemma} \label{lem:C}
 We have $C \subset \Pi_0$ and $C$ forms a $(d-1)$-dimensional box in $\Pi_0$.
 Therefore, a hypercube $B$ is a smallest hypercube enclosing $P$
 if and only if $B = \Ball(c, h/2)$ for some $c\in C$.
\end{lemma}
\begin{proof}
Let $B$ is any smallest axis-aligned hypercube $B$ that encloses $P$.
Since the side length of $B$ is $h$,
its center should lie on $\Pi_0$.
Hence, the set $C$ of centers of all smallest hypercubes that enclose $P$
is a subset of $\Pi_0$.
Furthermore, note that a hypercube $B$ encloses $P$ if and only if
$B$ encloses the smallest enclosing box $R$ for $P$.
This implies that $C$ forms a $(d-1)$-dimensional box in $\Pi_0$,
which may be degenerate to a box of lower dimension.
\end{proof}
In particular, if $d=3$, then $\Pi_0$ is the $x_1x_2$-plane,
and $C$ forms a rectangle in $\Pi_0$.
See \figurename~\ref{fig:C} for an illustration.

If we fix a center $c \in \Real^d$, then the minimum-width cubic shell $A^*(c)$
enclosing $P$ is uniquely determined as follows:
Since the outer cube $B$ of $A^*(c)$ should enclose all points of $P$,
we have $B = \Ball(c, r)$ with $r = \max_{p\in P} \|p-c\|$;
while the interior of the inner cube $B'$ of $A(c)$ should avoid $P$,
we have $B' = \Ball(c, r')$ with $r' = \min_{p \in P} \|p-c\|$.

For $d=2$, Gluchshenko et al.~\cite{ght-oafepramw-09} proved that
there always exists a minimum-width square annulus enclosing $P$
such that its center lies in $C$.
Here, we generalize this observation into higher dimensions.
\begin{lemma} \label{lem:C_centered}
 There exists a minimum-width hypercubic shell enclosing $P$
 centered at some $c \in C$.
\end{lemma}
\begin{proof}
Consider any minimum-width hypercubic shell $A = A^*(c^*)$
enclosing $P$ for $c^*\in \Real^d$.
That is, the center $c^*$ minimizes the width of $A^*(c)$ over all $c\in \Real^d$.
Let $B=\Ball(c^*, r)$ and $B'=\Ball(c^*, r')$ be its outer and inner hypercubes.
Note that $r = \max_{p\in P} \|p-c^*\|$ and $r' = \min_{p\in P} \|p-c^*\|$.
If $c^* \in C$, then we are done.

Suppose that $c^* \notin C$.
Then, by Lemma~\ref{lem:C}, $B$ is not a smallest enclosing hypercube for $P$.
Thus by Lemma~\ref{lem:smallest_enclosing_cube},
there is no pair of parallel facets of $B$
both of which contain a point of $P$.
On the other hand, for each pair of parallel facets of $B$,
at least one should contain a point of $P$ by our definition of $A= A^*(c^*)$.
Summarizing, there are exactly $d$ facets of $B$ containing a point of $P$
and no two of them are parallel.
Hence, there is a unique vertex $q$ of $B$ that is incident to these $d$ facets.
See \figurename~\ref{fig:C_centered}(a) for an illustration of the case of $d=2$.

\begin{figure}[tb]
\begin{center}
\includegraphics[width=.7\textwidth]{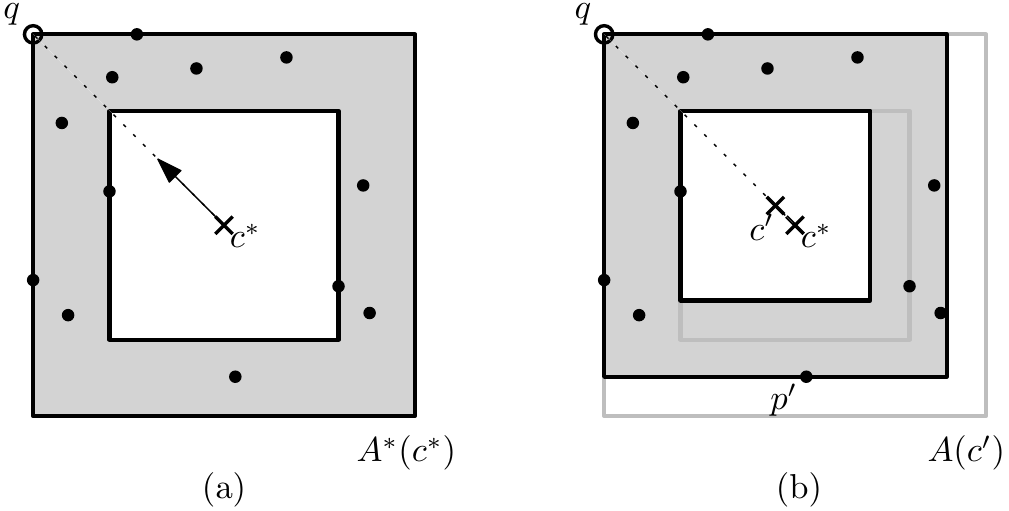}
\end{center}
\caption{An illustration to the proof of Lemma~\ref{lem:C_centered}.
 (a) A minimum-width square annulus $A = A^*(c^*)$ with $c^* \notin C$.
 Only $d=2$ facets (edges) of its outer square contain a point of $P$.
 (b) A new square annulus $A(c')$ whose width is the same as that of $A^*(c^*)$
 such that $d+1 = 3$ facets of its outer square $B(c')$ contains a point of $P$.
 }
\label{fig:C_centered}
\end{figure}

We now try to slide the center $c^*$ of the shell $A$ towards $q$.
For each $c$ on the line segment between $c^*$ and $q$,
we define a new hypercubic shell $A(c)$ such that
its outer hypercube is $B(c) = \Ball(c, r - \delta)$
and its inner hypercube is $B'(c) = \Ball(c, r' - \delta)$,
where $\delta = \|c-c^*\|$.
For any such $c$ with $\delta < r'$,
observe that $q$ is still a vertex of $B(c)$,
$B'(c) \subseteq B'(c^*)$ avoids the points in $P$ from its interior,
and the width of $A(c)$ is exactly $r - r'$, being the same as that of $A = A^*(c^*)$.
As $c$ continuously moves from $c^*$ towards $q$,
$B(c)$ encloses $P$ and thus $A(c)$ also encloses $P$
until another facet of $B(c)$ hits the $(d+1)$-st point $p'\in P$ at $c=c'$.
Hence, $A(c')$ is also another minimum-width hypercubic shell
enclosing $P$.
See \figurename~\ref{fig:C_centered}(b) for an illustration.

Finally, we show that $c' \in C$.
At $c = c'$,
observe that $B(c')$ has $d+1$ facets containing a point of $P$,
so two of the $d+1$ facets should be parallel.
Since $B(c')$ encloses $P$, we conclude that $B(c')$ is
a smallest enclosing hypercube for $P$ by Lemma~\ref{lem:smallest_enclosing_cube}.
Hence, we have $c' \in C$ by Lemma~\ref{lem:C}.
\end{proof}

This implies that we can now solve the problem by searching a center in $C\subset \Pi_0$.
For each $p\in P$ and $c\in \Pi_0$,
define $f_p(c) := \|c - p\|$ be the $L_\infty$ distance from $c$ to $p$.
Consider any minimum-width hypercubic shell $A^*(c)$ centered at $c\in C$.
Then, the radius of its outer hypercube is always fixed as $h/2$.
Hence, our problem of computing a minimum-width hypercubic shell enclosing $P$
is a bit simplified to the problem of maximizing the radius of inner hypercube:
 \[ \text{maximize } \min_{p\in P} f_p(c)  \text{ over } c \in C. \]
That is, we want to find a highest point in the lower envelope of
the functions $f_p$.
We define $\Phi(c) := \min_{p\in P} f_p(c)$ to be the lower envelope
of the functions $f_p$.

\begin{figure}[tb]
\begin{center}
\includegraphics[width=.85\textwidth]{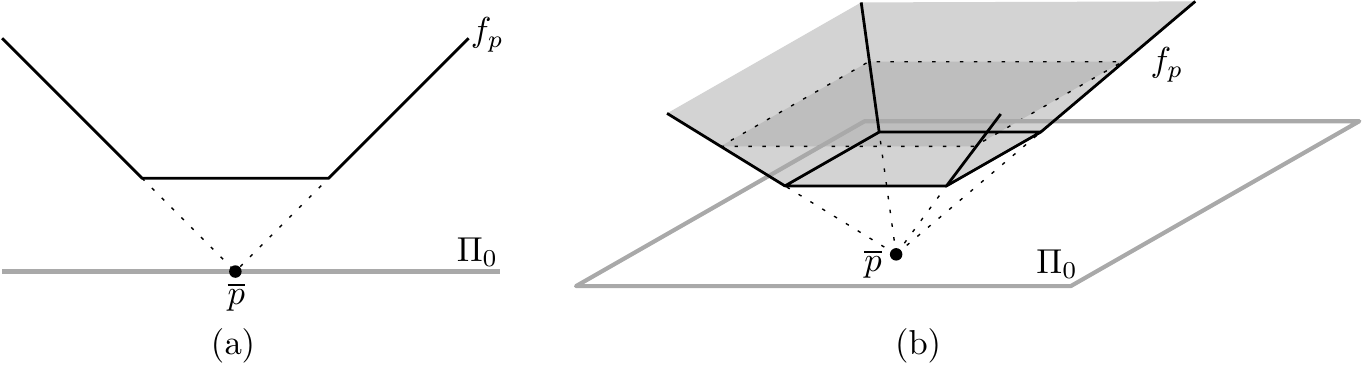}
\end{center}
\caption{An illustration to the graph of function $f_p$ for $p\in P$
for (a) $d=2$ and (b) $d=3$.
 }
\label{fig:function_f}
\end{figure}

Looking into the function $f_p$, it is defined on the $(d-1)$-dimensional
space $\Pi_0$ and
 \[ f_p(c) = \|p-c\| = \max_{i=1, \ldots d} | x_i(p) - x_i(c) |
    = \max\{ \max_{i=1, \ldots, d-1} |x_i(p) - x_i(c)|, |x_d(p)| \},\]
since $c \in \Pi_0$ and so $x_d(c) = 0$.
Observe that the first term $\max_{i=1, \ldots, d-1} |x_i(p) - x_i(c)|$
is the $L_\infty$ distance in a $(d-1)$-dimensional subspace,
while the second term $|x_d(p)|$ is a constant.
Thus, the graph $\{(c, z) \in \Pi_0\times\Real \mid z = f_p(c), c \in \Pi_0\}$
of $f_p$ is an $L_\infty$-cone cut by the hyperplane $\{z = |x_d(p)|\}$
parallel to $\Pi_0$.
See \figurename~\ref{fig:function_f} for an illustration.
From this graphical intuition, one can easily derive the following properties
of function $f_p$.
\begin{lemma} \label{lem:f_p}
 Let $p \in P$. Then, the following hold.
 \begin{enumerate}[(1)] \denseitems
  \item $f_p$ is convex.
  \item $f_p$ is piecewise linear with $2^{d-1} + 1$ patches, unless $x_d(p) = 0$.
  One of the patches forms a $(d-1)$-dimensional hypercube in $\Pi_0 \times \Real$, being parallel to $\Pi_0$.
  We call it the \emph{plateau} of $f_p$.
  \item Any point on the plateau of $f_p$ is a lowest point
  in the graph of $f_p$.
  That is, the global minimum of $f_p$ is attained at $c \in \Pi_0$
  if and only if $(c, f_p(c))$ is a point on the plateau of $f_p$,
  or equivalently, $f_p(c) = |x_d(p)|$.
 \end{enumerate}
\end{lemma}
\begin{proof}
From the fact that $f_p(c) = \max\{ \max_{i=1, \ldots, d-1} |x_i(p) - x_i(c)|, |x_d(p)| \}$, it is obvious that $f_p$ is convex.
The $L_\infty$ distance function $c \mapsto \max_{i=1, \ldots, d-1} |x_i(p) - x_i(c)|$ in $(d-1)$-dimensional space is convex and piecewise linear
with exactly $2^{d-1}$ patches.
If $|x_d(p)| \neq 0$, then the function $f_p(c)$ adds one more patch to it,
which is parallel to $\Pi_0$.
Thus, properties (1) and (2) are true.
This patch, called the plateau of $f_p$, forms the minimum of convex function $f_p$, so property (3) holds.
\end{proof}

From the above observations on the functions $f_p$,
we can discuss local maxima of their lower envelope $\Phi$.
\begin{lemma} \label{lem:phi_maxima}
 Let $c^* \in C$ be a local maximum of $\Phi$ on subdomain $C \subset \Pi_0$.
 Then, either
 \begin{enumerate}[(i)] \denseitems
  \item $\Phi(c^*) = f_p(c^*) = |x_d(p)|$ for some $p \in P$, or
  \item for some $0\leq d' \leq d-1$,
   $c^*$ lies in a face of $C$ of dimension $d'$ and
   there are $d'+1$ distinct points $p_1, \ldots, p_{d'+1} \in P$ such that
   $\Phi(c^*) = f_{p_1}(c^*) = \cdots = f_{p_{d'+1}}(c^*)$.
 \end{enumerate}
\end{lemma}
\begin{proof}
We make use of a general theorem on local maxima of the lower envelope
of convex functions, which was proved by Bae et al.~\cite{bko-gdpd-13},
stated as follows:
\begin{quote}
 (*) Let $d'$ be any positive integer.
 Let $\mathcal{F}$ be a finite family of real-valued convex functions
 defined on a convex subset $C'\subseteq \Real^{d'}$ and
 $g(c):= \min_{f \in\mathcal{F}} f(c)$ be their pointwise minimum.
 Suppose that $g$ attains a local maximum at $c^* \in C'$ and
 there are exactly $m \leq d'$ functions $f_1,\ldots,f_m \in\mathcal{F}$
 such that $f_i(c^*)=g(c^*)$ for each $i=1,\ldots,m$.
 Then, there exists a $(d'+1-m)$-flat\footnote{%
 A \emph{$d''$-flat} is an affine subspace of dimension $d''$.}
  $\varphi \subset \Real^{d'}$ through $c^*$
 such that $g$ is constant on $\varphi \cap U$
 for some neighborhood $U \subset \Real^{d'}$ of $c^*$ with $U\subset C'$.
\end{quote}
Informally speaking, the above theorem tells us that
if the number of functions that simultaneously appear on the lower envelope $g$
at a local maximum $c^*$ is not enough,
then $g$ is constant near $c^*$.
See~\cite{bko-gdpd-13} for its proof and discussion.

We apply the above theorem (*) to our situation.
Let $F$ be a face of $C$ of dimension $d'$ for $0 \leq d' \leq d-1$.
Note that, in particular, if $d' = d-1$, then $F$ is the interior of $C$.
Assume that $c^* \in F$ is a local maximum of $\Phi$ on $C$.
If $d'=0$, then $F$ is a vertex of $C$ and there must be at least one $p \in P$
such that $\Phi(c^*) = f_p(c^*)$, so we are done.
Thus, in the following, we assume $1 \leq d' \leq d-1$.

Now, assume that there are $m$ distinct points $p_1, \ldots, p_m \in P$
such that $\Phi(c^*) = f_{p_1}(c^*) = \cdots = f_{p_m}(c^*)$.
If $m \geq d'+1$, then this is case (ii) and we are done.
Suppose that $m \leq d'$.
Consider the restriction $f_p|_F$ of functions $f_p \colon \Pi_0 \to \Real$
to $F$, for each $p\in P$,
also the restriction $\Phi|_F$ of $\Phi$ to $F$.
Let $\mathcal{F} := \{f_p|_F \mid p \in P\}$.
Note that $\Phi|_F(c) = \min_{f \in \mathcal{F}} f(c)$
and $f_{p_i}|_F(c^*) = \Phi|_F(c^*)$ for each $i \in \{1, \dots, m\}$.
Since $c^*$ is a local maximum of $\Phi$, it is also a local maximum
of $\Phi|_F$ in $F$.
Hence, we can apply the theorem (*), concluding that
$\Phi|_F$ is constant near $c^*$.
This implies that every $f_{p_i}|_F$ must be constant near $c^*$
since $\Phi|_F(c^*) = f_{p_i}|_F(c^*)$.
From the properties of $f_p$ observed in Lemma~\ref{lem:f_p},
this is possible only if $f_{p_i}(c^*) = |x_d(p_i)|$.
So, this is case (i) of the lemma.
\end{proof}

\section{Algorithm for the Minimum-Width Cubic Shell} \label{sec:cubic_shell}
Let $P$ be a set of $n$ points in $\Real^3$.
In this section, we present an $O(n \log^2 n)$ time algorithm
that computes a minimum-width cubic shell enclosing $P$.

The function $f_p$ is piecewise linear of constant complexity,
defined on domain $\Pi_0$, which is a two-dimensional subspace.
Thus, one can apply an available machinery that computes the lower envelope
of the piecewise linear functions.
It was successful for the case of $d=2$; it is just computing
the lower envelope of line segments in the plane, and can be done
in $O(n \log n)$ worst-case time using a known algorithm,
as shown by Gluchshenko et al.~\cite{ght-oafepramw-09}.
However, for $d=3$, it takes $O(n^2 \alpha(n))$ time~\cite{egs-ueplf:aa-89}
to compute the envelope $\Phi$, and this is too much for us.

We suggest another approach
which does not explicitly compute the whole envelope $\Phi$.
Here, we consider the case of $d=3$.
Thus, $\Pi_0$ is the $x_1x_2$-plane and $C$ is a rectangle in $\Pi_0$.

For the purpose, we define for each $p\in P$ and $c\in \Pi_0$
 \[ \fproj_p(c) := \max\{ |x_1(p) - x_1(c)|, |x_2(p) - x_2(c)| \}. \]
Note that $f_p(c) = \max\{ |x_1(p) - x_1(c)|, |x_2(p) - x_2(c)|, |x_3(p)| \}$.
Thus, $\fproj_p(c)$ is the $L_\infty$ distance between $c\in \Pi_0$
and the orthogonal projection of $p$ onto $\Pi_0$,
so basically $L_\infty$ distance between two points in a plane.
Let $w^*$ be the width of a minimum-width cubic shell enclosing $P$.
As discussed above, we have
 \[ w^* = h/2 - \max_{c \in C} \Phi(c),\]
where $h$ is the longest side length of
the smallest enclosing box $R$ for $P$ as defined above.
Let $r^* = \max_{c\in C} \Phi(c)$ and $c^* \in C$ be such that $\Phi(c^*) = r^*$.
Since $\fproj_p(c) = f_p(c)$ unless $f_p(c) = |x_3(p)|$,
Lemma~\ref{lem:phi_maxima} implies the following.
\begin{lemma} \label{lem:conf_3d}
One of the following cases (i) and (ii) holds:
\begin{enumerate}[(i)] \denseitems
 \item $r^* = |x_3(p)|$ for some $p\in P$.
 \item $c^*$ is either
  \begin{enumerate}[(a)]\denseitems
   \item a point in the interior of $C$ such that
   $r^* = \fproj_{p_1}(c^*) = \fproj_{p_2}(c^*) = \fproj_{p_3}(c^*)$
   for some $p_1, p_2, p_3 \in P$,
   \item a point on an edge of $C$ such that
   $r^* = \fproj_{p_1}(c^*) = \fproj_{p_2}(c^*)$
   for some $p_1, p_2 \in P$, or
   \item a vertex of $C$.
  \end{enumerate}
\end{enumerate}
\end{lemma}
\begin{proof}
Recall that $c^* \in C$ maximizes $\Phi$ over $C$ and $r^* = \Phi(c^*)$,
so $c^*$ is a local maximum of $\Phi$ in $C$.
Hence, we can apply Lemma~\ref{lem:phi_maxima} for $c^*$.

Suppose that we are not in case (i) of the lemma, in which we have
$r^* = |x_3(p)|$ for some $p\in P$.
In other words, we suppose that $r^* \neq |x_3(p)|$ for all $p\in P$.
Note that this implies that $f_p(c^*) = \fproj_p(c^*)$ for all $p\in P$,
as discussed above.
This also excludes case (i) of Lemma~\ref{lem:phi_maxima},
so this should be case (ii) of Lemma~\ref{lem:phi_maxima}.
Specifically, it holds that
for $0\leq d' \leq 2$, $c^*$ lies in a $d'$-face of $C$ and
there are $d'+1$ distinct points $p_1, \dots, p_{d'+1} \in P$ such that
 \[r^* = \Phi(c^*) = f_{p_1}(c^*) = \cdots = f_{p_{d'+1}}(c^*).\]

There are three cases according to the dimension $d'$ of the face in which $c^*$ lies.
\begin{enumerate}[(a)] \denseitems
\item
If $c^*$ lies in a $2$-face of $C$, or the interior of $C$, then
there are three points $p_1, p_2, p_3 \in P$ such that
$r^* = f_{p_1}(c^*) = f_{p_2}(c^*) = f_{p_3}(c^*)$.
Since $f_p(c^*) = \fproj_p(c^*)$ for all $p\in P$,
we have $r^* = \fproj_{p_1}(c^*) = \fproj_{p_2}(c^*) = \fproj_{p_3}(c^*)$.
\item
If $c^*$ lies in a $1$-face, or an edge, of $C$, then
there are two points $p_1, p_2 \in P$ such that
$r^* = f_{p_1}(c^*) = f_{p_2}(c^*)$.
Since $f_p(c^*) = \fproj_p(c^*)$ for all $p\in P$,
we have $r^* = \fproj_{p_1}(c^*) = \fproj_{p_2}(c^*)$.
\item
If $c^*$ lies in a $0$-face, then $c^*$ is a vertex of $C$.
\end{enumerate}
Hence, the lemma follows.
\end{proof}

Our algorithm computes $r^* = \max_{c\in C} \Phi(c)$ and
a corresponding center $c^*$ such that $r^* = \Phi(c^*)$
by separately handling two cases (i) and (ii) of Lemma~\ref{lem:conf_3d}.
For case (i), let $r^*_1$ be the largest value in $\{ |x_3(p)| \mid p\in P\}$
such that there exists a cubic shell of width $r^*_1$ and center $c^*_1\in C$
that encloses $P$.
If the solution $r^*$ falls in case (i), then it should hold that $r^* = r^*_1$.

For case (ii), any point $c \in C$ is called a \emph{candidate center}
if it satisfies the condition of case (ii); more precisely, if $c$ is either
  \begin{enumerate}[(a)]\denseitems
   \item a point in the interior of $C$ such that
   $\Phi(c) = \fproj_{p_1}(c) = \fproj_{p_2}(c) = \fproj_{p_3}(c)$
   for some $p_1, p_2, p_3 \in P$,
   \item a point on an edge of $C$ such that
   $\Phi(c) = \fproj_{p_1}(c) = \fproj_{p_2}(c)$
   for some $p_1, p_2 \in P$, or
   \item a vertex of $C$.
  \end{enumerate}
Let $Q$ be the set of all candidate centers, and
let $r^*_2 := \max_{c \in Q} \Phi(c)$ and
$c^*_2 \in Q$ be such that $r^*_2 = \Phi(c^*_2)$.
If the solution $r^*$ and $c^*$ does not fall in case (i),
then we will have $r^* = r^*_2$.

Our algorithm thus computes $r^*_1$ and $r^*_2$ and then
$r^* = \max\{r^*_1, r^*_2\}$ by Lemma~\ref{lem:conf_3d}.
Hence, we are done by reporting $r^* = \max\{r^*_1, r^*_2\}$
and its corresponding center and cubic shell.
Note that the width of the minimum-width shell is $h/2 - r^*$.

In the following, we describe how to handle each case and compute
$r^*_1$ and $r^*_2$.

\subsection{Case (i)}
Note that there are only $n$ candidate values $\{ |x_3(p)| \mid p\in P\}$
for $r^*_1$.
Here, we consider the following decision problem:
\begin{quote}
 \textit{given a real $w\geq 0$, is there a cubic shell $A$ enclosing $P$
 with width at most $w$ and center in $C$?}
\end{quote}
This is equivalent to deciding if the \emph{sublevel set} $U(w)$ of $\Phi$
covers $C$,
where
 \[U(w) := \{ c\in \Pi_0 \mid \Phi(c) < w \} \]
that is, whether or not $C \subseteq U(w)$.

For a given real number $w$ and any $c \in \Pi_0$,
$\Phi(c) < w$ if and only if there exists a point $p\in P$
such that $f_p(c) < w$.
Since $f_p(c) = \|c - p\|$, the above condition is again equivalent to
$c \in \intr \Ball(p, w)$ for some $p\in P$ or
$c \in \bigcup_{p\in P} \intr \Ball(p, w)$,
where $\intr \Ball(p,w)$ denotes the interior of $\Ball(p,w)$.
Hence, $U(w)$ is indeed the intersection of the union
$\bigcup_{p \in P} \intr \Ball(p, w)$ of $n$ cubes by $\Pi_0$.

Let $\Ball_0(p, r) := \Ball(p, r) \cap \Pi_0$ be
the intersection of the $L_\infty$ ball $\Ball(p, r)$ by $\Pi_0$.
We then have
 \[ U(w) = \bigcup_{p \in P} \intr \Ball_0(p, w).\]
Note that $\Ball_0(p,w)$ is either empty if $|x_3(p)| \geq w$, or
a square of radius $w$.

After specifying $\Ball_0(p, w)$ for each $p\in P$,
we can explicitly compute the union $U(w)$ of squares of equal radius $w$.
It is well known that the complexity of the union of $n$ squares is $O(n)$~\cite{bsty-vdhdcpdf-98}, and
one can compute it in $O(n \log n)$ time by a standard plane-sweep algorithm.
We then intersect $U(W)$ by $C$.
If there is a point $c \in C$ such that $c \notin U(w)$,
then we have $\Phi(c) \geq w$ and thus the cubic shell $A^*(c)$ centered at $c$
has width at most $w$, so we report that there exists a cubic shell of width $w$
enclosing $P$.
Otherwise, if $C \subseteq U(w)$, then there is no such shell.

Thus, we conclude the following.
\begin{lemma} \label{lemma:decision_3d}
 Given a set $P$ of $n$ points in $\Real^3$ and a real $w \geq 0$,
 we can decide if there exists a cubic shell enclosing $P$ of width $w$
 in $O(n \log n)$ time in the worst case.
 If exists, such a cubic shell can be output in the same time bound.
\end{lemma}

After sorting $\{ |x_3(p)|  \mid p \in P\}$ in $O(n \log n)$ time,
we can find the biggest value $r^*_1$ for which the above decision algorithm
returns ``yes'' in $O(n \log^2 n)$ time by a binary search.
Such a point $c^*_1 \in C$ that $r^*_1 = \Phi(c^*_1)$ can also be found
in the same time bound.

\subsection{Case (ii)}
Next, we describe how to compute $r^*_2$ and $c^*_2$.
As defined above, $Q\subset C$ is the set of all candidate centers.
Again, recall that $\fproj_p(c)$ is equivalent to the $L_\infty$ distance
in the plane $\Pi_0$ between the projection of $p$ onto $\Pi_0$ and
a point $c\in \Pi_0$.
This means that each candidate center $c$ is, unless it is a vertex of $C$,
a point on the locus of equidistant points from two or more
points in on the plane $\Pi_0$ under the $L_\infty$ distance.
This naturally suggests an application of the $L_\infty$ Voronoi diagram
in the plane $\Pi_0$.

For each $p\in P$, let $\proj{p}$ be the orthogonal projection of $p$
onto the plane $\Pi_0$, and $\proj{P}:= \{\proj{p} \mid p\in P\}$.
Let $\VD(\proj{P})$ be the $L_\infty$ Voronoi diagram for points $\proj{P}$
on $\Pi_0$, that is,
the decomposition of $\Pi_0$ into vertices, edges, and cells,
each of which is the set of points having a common set of nearest points
in $\proj{P}$ under the $L_\infty$ distance.
It is well known that $\VD(\proj{P})$ consists of $O(n)$ vertices, edges, and faces,
and can be computed in $O(n \log n)$ time~\cite{ }.
In particular, we have the following:
\begin{enumerate}[(a)] \denseitems
 \item A point $c \in \Pi_0$ is a vertex of $\VD(\proj{P})$ if and only if
we have three nearest points $\proj{p_1}, \proj{p_2}, \proj{p_3} \in \proj{P}$
so that $\fproj_{p_1}(c) = \fproj_{p_2}(c) = \fproj_{p_3}(c)$.
 \item A point $c \in \Pi_0$ lies on an edge of $\VD(\proj{P})$ if and only if
 we have exactly two nearest points $\proj{p_1}, \proj{p_2} \in \proj{P}$
so that $\fproj_{p_1}(c) = \fproj_{p_2}(c)$.
\end{enumerate}
This gives us a necessary condition of candidate centers.
\begin{lemma} \label{lem:candidate_center_3d}
 Let $c \in Q$ be a candidate center.
 Then, $c$ is either a vertex of $\VD(\proj{P})$,
 an intersection of an edge of $\VD(\proj{P})$ and an edge of $C$,
 or a vertex of $C$.
\end{lemma}
\begin{proof}
By definition, any candidate center $c$ should satisfy either
(a) there are three points $p_1, p_2, p_3 \in P$ such that
$\fproj_{p_1}(c) = \fproj_{p_2}(c) = \fproj_{p_3}(c)$,
(b) $c$ is a point on an edge of $C$ and there are two points $p_1, p_2 \in P$
such that $\fproj_{p_1}(c) = \fproj_{p_2}(c)$, or
(c) $c$ is a vertex of $C$.
From the property of the Voronoi diagram $\VD(\proj{P})$ discussed above,
in case (a), $c$ is a vertex of $\VD(\proj{P})$;
in case (b), $c$ is also a point on an edge of $\VD(\proj{P})$.
Hence, the lemma is proved.
\end{proof}

Now, we are ready to describe our algorithm computing $r^*_2$:
We first compute $\VD(\proj{P})$ and then compute the intersection points
between edges of $\VD(\proj{P})$ and edges of $C$.
Initially, we let $Q$ include all vertices of $\VD(\proj{P})$,
all intersection points between an edge of $\VD(\proj{P})$ and an edge of $C$,
and all vertices of $C$.
By Lemma~\ref{lem:candidate_center_3d}, $Q$ contains all candidate centers.
For each $c \in Q$, test if $f_p(c) = \fproj_p(c)$
for every nearest point $\proj{p}$ from $c$ among $\proj{P}$.
This test can be done in $O(1)$ time since $\VD(\proj{P})$ stores
nearest points for each vertex, edge, and cell.
If the test is passed, $c$ is a candidate center by definition
and so we keep $c$ in $Q$;
otherwise, we discard $c$ and remove $c$ from $Q$.
Now, $Q$ consists of all candidate centers.
Note that if $c$ is a candidate center, it holds that
$\Phi(c) = f_p(c) = \fproj_p(c)$ for each nearest point $\proj{p}\in \proj{P}$
from $c$.
We then pick a candidate center $c^*_2 \in Q$ such that
$\Phi(c^*_2) = \max_{c \in Q} \Phi(c) = r^*_2$.
All the effort to compute $r^*_2$ and $c^*_2$ is bounded by
$O(n \log n)$ time.

\bigskip

Summarizing, we handle two cases (i) and (ii) of Lemma~\ref{lem:conf_3d}
separately, computing $r^*_1$ and $r^*_2$, and choose the bigger one as $r^*$.
Then, a minimum-width cubic shell enclosing $P$ is obtained
from the corresponding center and the radii $h/2$ and $r^*$
of its outer and inner cubes.

\begin{theorem} \label{thm:cubic_shell}
 Let $P$ be a set of $n$ points in $\Real^3$.
 A minimum-width cubic shell enclosing $P$ can be computed in
 $O(n \log^2 n)$ time in the worst case.
\end{theorem}

\section{Minimum-Width Hypercubic Shell} \label{sec:hypercubic_shell}
Our approach for cubic shells in $\Real^3$ easily extends to
hypercubic shells in $\Real^d$ for $d > 3$.
In this section, let $d > 3$ be a constant.

As done for $d=3$, we define
 \[ \fproj_p(c) := \max_{i=1, \ldots, d-1} |x_i(p) - x_i(c)|. \]
Note that $f_p(c) = \max\{ \fproj_p(c), |x_d(p)| \}$ and
$f_p(c) = \fproj_p(c)$ unless $f_p(c) = |x_d(p)|$.
Thus, $\fproj_p(c)$ is the $L_\infty$ distance between $c\in \Pi_0$
and the orthogonal projection of $p$ onto $\Pi_0$,
so the $L_\infty$ distance between two points in the $(d-1)$-dimensional space.
Let $w^*, r^*$ be defined as above.
So, we have an analogue of Lemma~\ref{lem:conf_3d}.
\begin{lemma} \label{lem:conf_highd}
\begin{enumerate}[(i)] \denseitems
 \item $r^* = |x_d(p)|$ for some $p\in P$, or
 \item $c^*$ is a point in a face of $C$ of dimension $d'$ with $0 \leq d' \leq d-1$ and $r^* = \fproj_{p_1}(c^*) = \cdots = \fproj_{p_{d'+1}}(c^*)$
 for $d'+1$ distinct points $p_1, \ldots, p_{d'+1} \in P$.
\end{enumerate}
\end{lemma}
\begin{proof}
The proof is almost identical to that of Lemma~\ref{lem:conf_3d}.
Since $c^*$ is a local maximum of $\Phi$ in $C$,
we can apply Lemma~\ref{lem:phi_maxima} for $c^*$.

Suppose that we are not in case (i) of the lemma, in which we have
$r^* = |x_d(p)|$ for some $p\in P$.
In other words, we suppose that $r^* \neq |x_d(p)|$ for all $p\in P$.
Note that this implies that $f_p(c^*) = \fproj_p(c^*)$ for all $p\in P$,
as discussed above.
This also excludes case (i) of Lemma~\ref{lem:phi_maxima},
so this should be case (ii) of Lemma~\ref{lem:phi_maxima}.
Specifically, it holds that
for $0\leq d' \leq d-1$, $c^*$ lies in a $d'$-face of $C$ and
there are $d'+1$ distinct points $p_1, \dots, p_{d'+1} \in P$ such that
 \[r^* = \Phi(c^*) = f_{p_1}(c^*) = \cdots = f_{d'+1}(c^*).\]
Since we have $f_p(c^*) = \fproj_p(c^*)$ for all $p\in P$,
this implies that
$r^* = \fproj_{p_1}(c^*) = \cdots = \fproj_{p_{d'+1}}(c^*)$,
as claimed.
\end{proof}

As done for $d=3$, our algorithm computes $r^* = \max_{c\in C} \Phi(c)$
and a corresponding center $c^* \in C$ such that $r^* = \Phi(c^*)$
by separately handling two cases (i) and (ii).
Each case is also handled similarly:
we define $r^*_1$ and $r^*_2$ analogously.
In particular, a point $c \in C$ is called a \emph{candidate center}
if $c$ lies in a $d'$-face of $C$ for $0\leq d' \leq d-1$
and there are $d'+1$ distinct points $p_1, \ldots, p_{d'+1} \in P$
such that $\Phi(c) = \fproj_{p_1}(c) = \cdots = \fproj_{p_{d'+1}}(c)$.

For our algorithm for $d>3$, an essential tool is again
the $L_\infty$ Voronoi diagram $\VD(\proj{P})$ in $d-1$ dimensional space $\Pi_0$.
The diagram $\VD(\proj{P})$ decomposes $\Pi_0$ into faces of dimension
$d' \in \{0, \ldots, d-1\}$ such that
each $d'$-face $F$ of $\VD(\proj{P})$ is the maximal set of points $c \in \Pi_0$
having a common set $N(F)$ of $d - d'$ nearest points in $\proj{P}$.
Fortunately, Boissonat et al.~\cite{bsty-vdhdcpdf-98} proved the following:
\begin{lemma}[Boissonat et al.~\cite{bsty-vdhdcpdf-98}] \label{lem:vd_highd}
 The $L_\infty$ Voronoi diagram of $n$ points in $d-1$ dimension
 has complexity $O(n^{\lfloor d/2 \rfloor}$
 and can be computed in $O(n^{\lfloor d/2 \rfloor} \log^{d-2} n)$ expected time.
\end{lemma}

We again handle each case separately.
\subsection{Case (i)}
We again consider the decision problem, and solve it
by testing $C \subseteq U(w)$.
The only difference is that $U(w)$ is now the union of
$(d-1)$-dimensional hypercubes $\Ball_0(p, w)$ for $p\in P$ in $\Pi_0$.

It is known by Boissonat et al.~\cite{bsty-vdhdcpdf-98} that the union of $n$ hypercubes of equal radius in $d-1$ dimension
has complexity $O(n^{\lfloor (d-1)/2\rfloor})$ for $d \geq 3$.
We can compute the union $U(w)$ of hypercubes in $\Pi_0$ by using
the $(d-1)$-dimensional $L_\infty$ Voronoi diagram.
\begin{lemma}\label{lem:union_hypercubes}
 Let $S$ be a set of $m$ hypercubes of equal radius in $d-1$ dimensional space.
 Then, their union can be computed in
 $O(m^{\lfloor d/2 \rfloor} \log^{d-2} m)$ expected time.
\end{lemma}
\begin{proof}
Let $w$ be the radius of hypercubes in $S$, and $P'$ be the set of centers
of hypercubes in $S$.
Let $U$ be their union $\bigcup_{B \in S} B$.
We compute $U$ using the $L_\infty$ Voronoi diagram $\VD(P')$.
Note that,
for each point $c\in U$, it holds that
$\min_{p \in P'} \|c - p\| < w$,
since each $B\in S$ is $\Ball(p, w)$ for some $p \in P'$.

We first compute the $L_\infty$ Voronoi diagram $\VD(P')$.
This takes $O(m^{\lfloor d/2 \rfloor} \log^{d-2} m)$ expected time
by Lemma~\ref{lem:vd_highd}.
Then, for each face $F$ of $\VD(P')$, we compute the set $U(F)$
of points $c \in F$ such that $\| c - p\| < w$ for every $p \in N(F)$.
Note that the set $N(F)$ of common nearest points for face $F$ consists of
$d-d'$ points, if $F$ is a $d'$-face, and $U(F)$ is just the intersection
\[ U(F) = F \cap \bigcap_{p\in N(F)} \Ball_0(p, w)\]
of $d-d'$ hypercubes and the face $F$.
Since the complexity of $\VD(P')$ is $O(m^{\lfloor d/2 \rfloor})$,
this iteration is done in time $O(m^{\lfloor d/2 \rfloor})$.
Since the faces of $\VD(P')$ form a (disjoint) decomposition of the space,
we have $U = \bigcup_F U(F)$.
Thus, we can compute the union $U$ in the claimed time.
\end{proof}

After specifying $\Ball_0(p, w)$ for each $p\in P$,
we collect at most $n$ hypercubes of $d-1$ dimension and compute their union
$U(w)$ by the algorithm of Lemma~\ref{lem:union_hypercubes}.
Then, we intersect $U(w)$ by $C$.
Since the complexity of $U(w)$ is bounded by $O(n^{\lfloor (d-1)/2\rfloor})$,
this can be also done in the same time bound.
If there is a point $c \in C$ such that $c \notin U(w)$,
then we have $\Phi(c) \geq w$ and thus the hypercubic shell $A^*(c)$
centered at $c$ has width at most $w$,
so we report that there exists a hypercubic shell of width $w$
enclosing $P$.
Otherwise, if $C \subseteq U(w)$, then there is no such shell.

Thus, we conclude the following.
\begin{lemma} \label{lemma:decision_highd}
 Let $d \geq 4$ be a constant.
 Given a set $P$ of $n$ points in $\Real^d$ and a real $w \geq 0$,
 we can decide if there exists a hypercubic shell enclosing $P$ of width $w$
 in $O(n^{\lfloor d/2 \rfloor} \log^{d-2} n)$ expected time.
 If exists, such a hypercubic shell can be output in the same time bound.
\end{lemma}

After sorting $\{ |x_d(p)|  \mid p \in P\}$ in $O(n \log n)$ time,
we can find the biggest value $r^*_1$ for which the above decision algorithm
returns ``yes'' in $O(n^{\lfloor d/2 \rfloor} \log^{d-1} n)$ time
by a binary search.
Such a point $c^*_1 \in C$ that $r^*_1 = \Phi(c^*_1)$ can also be found
in the same time bound.

\subsection{Case (ii)}
In order to compute $r^*_2$ and $c^*_2$ for $d > 3$,
we show an analogous lemma of Lemma~\ref{lem:candidate_center_3d}.
\begin{lemma} \label{lem:candidate_center_highd}
 Let $c \in Q$ be a candidate center.
 Then, $c$ is an intersection of a $d'$-face of $C$ and
 and a $(d-d'-1)$-face of $\VD(\proj{P})$ for some $0\leq d' \leq d-1$.
\end{lemma}
\begin{proof}
By definition, any candidate center $c$ is a point on a $d'$-face of $C$
for $0 \leq d' \leq d-1$
such that there are $d'+1$ distinct points $p_1, \ldots, p_{d'+1} \in P$
such that $\Phi(c) = \fproj_{p_1}(c) = \cdots = \fproj_{p_{d'+1}}(c)$.
Since $\Phi(c) \leq \min_{p\in P} \fproj_p(c)$,
those $d'+1$ points are all nearest points in $\proj{P}$ from $c$.
From the property of the Voronoi diagram $\VD(\proj{P})$,
this implies that $c$ lies in a face of $\VD(\proj{P})$
of dimension $d-(d'+1) = d-d'-1$.
Hence, $c$ is an intersection point of a $d'$-face of $C$ and
 and a $(d-d'-1)$-face of $\VD(\proj{P})$ for some $0\leq d' \leq d-1$.
\end{proof}

Our algorithm thus computes $\VD(\proj{P})$ and intersects it with $C$.
Initially, we let $Q$ be the set of all intersection points
between a $d'$-face of $C$ and a $(d-d'-1)$-face of $\VD(\proj{P})$
for $0\leq d' \leq d-1$.
Since $\VD(\proj{P})$ consists of $O(n^{\lfloor d/2 \rfloor})$ faces
(Lemma~\ref{lem:vd_highd}), $Q$ consists of at most $O(n^{\lfloor d/2 \rfloor})$
points.
By Lemma~\ref{lem:candidate_center_highd}, it is guaranteed that
$Q$ contains all candidate centers.
We then test each $c \in Q$ if $f_p(c) = \fproj_p(c)$
for all nearest points $\proj{p} \in \proj{P}$ from $c$.
This can be done in $O(d) = O(1)$ time
by storing the face $F$ of $\VD(\proj{P})$ that contains $c$
and the set $N(F)$ of its nearest points.
If the test is passed, $c$ is a candidate center;
otherwise, we discard $c$ and remove $c$ from $Q$.
Now, $Q$ consists of only candidate centers.
Note that if $c$ is a candidate center and $F$ is the face of $\VD(\proj{P})$
such that $c\in F$, it holds that
$\Phi(c) = f_p(c) = \fproj_p(c)$ for each $p\in N(F)$.
Thus, we can find $c^*_2$ and $r^*_2$ simply taking the maximum
$r^*_2 = \max_{c \in Q} \Phi(c) = \Phi(c^*_2)$.
The time consumed in this process is bounded by
$O(n^{\lfloor d/2\rfloor} \log^{d-2} n)$ expected time
for computing $\VD(\proj{P})$ by Lemma~\ref{lem:vd_highd}.

Finally, we conclude the following.
\begin{theorem} \label{thm:hypercubic_shell}
 Let $d\geq 4$ be a constant integer and $P$ be a set of $n$ points in $\Real^d$.
 Then, a minimum-width hypercubic shell enclosing $P$ can be computed in
 $O(n^{\lfloor d/2 \rfloor} \log^{d-1} n)$ expected time.
\end{theorem}

\section{Concluding Remarks} \label{sec:conclusions}

We addressed the minimum-width cubic and hypercubic shell problem
in high dimension, generalizing the square annulus problem.
Our algorithm runs in $O(n \log^2 n)$ worst-case time for the cubic shell
and $O(n^{\lfloor d/2 \rfloor} \log^{d-1} n)$ expected time
for the hypercubic shell in $\Real^d$ for $d \geq 4$.
It would be worth mentioning that the currently best time bound
$O(n^{\lfloor d/2 \rfloor} \log^{d-1} n)$ holds for any $d \geq 2$.
Theorems~\ref{thm:cubic_shell} and~\ref{thm:hypercubic_shell},
together with the result in~\cite{ght-oafepramw-09},
are summarized into the following corollary.
\begin{corollary}
 Let $d \geq 2$ be any constant integer, and $P$ be a set of $n$ points
 in $\Real^d$.
 Then, a minimum-width hypercubic shell enclosing $P$ can be computed
 in $O(n^{\lfloor d/2 \rfloor} \log^{d-1} n)$ time.
\end{corollary}

There are several open questions.
In particular for $d=3$, our algorithm runs in $O(n \log^2 n)$ time.
Is it possible to reduce the time bound to $O(n \log n)$?
As Gluckshenko et al.~\cite{ght-oafepramw-09} proved a lower bound
of $\Omega(n \log n)$ for $d=2$, the same lower bound applies to
the case of $d \geq 3$.

A bottleneck of our algorithm for $d=3$ is the decision algorithm
that takes $O(n \log n)$ time and the binary search using it.
One could try to apply the parametric search technique,
while it seems nontrivial to devise a proper parallel algorithm.

Another interesting question would be about the lower envelope of
functions $f_p$.
What is the correct complexity of the lower envelope $\Phi$ of functions $f_p$?
We tried to obtain a nontrivial upper bound, i.e., a subquadratic bound for $d=3$,
on the complexity of $\Phi$, but failed.
Note that the corresponding minimization diagram on $\Pi_0$ coincides
the intersection of a $d$-dimensional $L_\infty$ Voronoi diagram
by an axis-aligned hyperplane $\Pi_0$.



%

{
\bibliographystyle{abbrv}
\bibliography{annuli}
}

\end{document}